\documentclass[preprint,3p]{elsarticle} 
\usepackage{amssymb}
\usepackage{amsfonts}
\usepackage{amsmath}
\usepackage{amsthm}
\usepackage{algorithmic}
\usepackage{algorithm}
\usepackage{booktabs}
\usepackage{url}
\usepackage{multirow}
\usepackage{IEEEtrantools}
\usepackage{graphicx,color}
\usepackage[bf,SL,BF]{subfigure}
\usepackage{natbib}\biboptions{authoryear}
\usepackage{color,xcolor}
\usepackage[colorlinks, citecolor=blue, linkcolor=black]{hyperref}
\usepackage{enumerate}
\usepackage[title]{appendix}
\usepackage{lscape}
\usepackage{makecell}
\usepackage{colortbl}

\newcommand{\bM}{\mathbf{M}}\newcommand{\bX}{\mathbf{X}}\newcommand{\bx}{\mathbf{x}}

\newcommand{\cN}{\mathcal{N}}

\newcommand{\bwM}{\widehat{\bM}}\newcommand{\bwSig}{\widehat{\bSig}}\newcommand{\bwmu}{\widehat{\bmu}}

\newcommand{\bmu}{\boldsymbol{\mu}} 
\newcommand{\bSig}{\boldsymbol{\Sigma}}

\newcommand\refe[1]{(\ref{#1})}\newcommand\refp[1]{Proposition~\ref{#1}}
\newcommand\reff[1]{Fig.~\ref{#1}}\newcommand\reft[1]{Table~\ref{#1}}\newcommand\refs[1]{Sec.~\ref{#1}} 

\theoremstyle{plain} \newtheorem{prop}{Proposition} 

\def\R{{\mathbb R}}
\def\etr{\mbox{etr}}
\def\tr{\mbox{tr}}
\def\vc{\mbox{vec}}

\def\s2{\sigma^2}

\begin{document}	
		\title{Matrix Healy Plot: A Practical Tool for Visual Assessment of Matrix-Variate Normality}
  \begin{frontmatter}
        \author[J. Zhao]{Fen Jiang}
 \author[J. Zhao]{Jianhua Zhao\corref{cor1}}\cortext[cor1]{Corresponding author.} \ead{jhzhao.ynu@gmail.com}
 \author[C. Shang]{ Changchun Shang}
 \author[J. Zhao]{Xuan Ma}
 \author[J. Zhao]{Yue Wang}
 \author[J. Zhao]{Ye Tao}

 \address[J. Zhao]{School of Statistics and Mathematics, Yunnan University of Finance and Economics, Kunming, 650221, China}
 \address[C. Shang]{School of Mathematics and Statistics, Guilin University of Technology, Guilin, 541004, China}
	
	\begin{abstract}
		Matrix-valued data, where each observation is represented as a matrix, frequently arises in various scientific disciplines. Modeling such data often relies on matrix-variate normal distributions, making matrix-variate normality testing crucial for valid statistical inference. Recently, the Distance-Distance (DD) plot has been introduced as a graphical tool for visually assessing matrix-variate normality. However, the Mahalanobis squared distances (MSD) used in the DD plot require vectorizing matrix observations, restricting its applicability to cases where the dimension of the vectorized data does not exceed the sample size. To address this limitation, we propose a novel graphical method called the Matrix Healy (MHealy) plot, an extension of the Healy plot for vector-valued data. This new plot is based on more accurate matrix-based MSD that leverages the inherent structure of matrix data. Consequently, it offers a more reliable visual assessment. Importantly, the MHealy plot eliminates the sample size restriction of the DD plot and hence more applicable to matrix-valued data. Empirical results demonstrate its effectiveness and practicality compared to the DD plot across various scenarios, particularly in cases where the DD plot is not available due to limited sample sizes.
	\end{abstract}
	
	
	\end{frontmatter}
	\section{Introduction}\label{sec:intr} 
	Matrix-valued data have garnered significant attention across various fields due to their natural representation of complex relationships. In environmental studies, for instance, variables of interest such as temperature, humidity, or air quality are often recorded across geographical regions, where the data can be structured as matrices with dimensions corresponding to latitude and longitude. Similarly, matrix-valued data are prevalent in two-dimensional digital imaging, brain surface mapping, and spatio-temporal data studies \citep{tomarchio2022mixtures, zhang2023covariance}. These data structures provide an intuitive and efficient framework for capturing multidimensional dependencies and enabling advanced analyses. Recent years have witnessed increasing interest in the analysis of matrix-valued data, with substantial contributions from studies, such as \cite{gallaugher2017matrix, gallaugher2018finite, gallaugher2024clustering,kong2020l2rm, chen2023statistical, zhao2023-rfpca,shang2024bfa, zhang2024modeling,zhang2025balanced}.

Modeling such data of size $c\times r \times N$, often relies on matrix-variate normal (MVN) distributions, an extension of the multivariate normal distribution that effectively captures dependencies across rows and columns. This property makes the MVN distribution particularly well suited for analyzing such data. Similar to the importance of the multivariate normal distribution in traditional multivariate analysis, assessing whether the underlying data conform to matrix-variate normality is crucial. 
While methods for testing multivariate normality are well-established \citep{ramzan2013evaluating,korkmaz2014mvn,kim2018likelihood}, research on matrix-variate normality testing remains in its infancy, with relatively few studies available. A two-phase approach is commonly employed to assess matrix-variate normality. In the first phase, multivariate normality is verified using established techniques. In the second phase, matrix-variate normality is tested by examining whether the vectorized data exhibits a Kronecker-structured covariance matrix. Likelihood ratio tests (LRTs) are widely used for evaluating the presence of a Kronecker-structured covariance matrix \citep{lu2005likelihood}. However, both phases of the approach become infeasible when the dimension $cr$ of the vectorized data of size $cr \times N$ exceeds the sample size $N$, significantly limiting the applicability of the two-phase approach.

Graphical tools, such as the QQ plot \citep{easton1990multivariate} and the Healy plot \citep{healy1968multivariate}, are extensively used to assess multivariate normality due to their visual appeal and interpretability. Building upon this, graphical methods tailored specifically for matrix-variate normality provide critical insights by identifying deviations from the assumed distribution. These tools enable researchers to evaluate model adequacy and enhance the robustness of statistical inferences. Recently, the Distance-Distance (DD) plot has been proposed as a graphical tool to visually assess matrix-variate normality \citep{povcuvca2023visual}. However, the Mahalanobis squared distances (MSD) used in the DD plot require vectorizing matrix observations, which renders maximum likelihood estimation (MLE) infeasible when the dimension $cr$ of the vectorized data exceeds the sample size $N$\citep{dutilleul1999mle}. This limitation restricts the DD plot’s applicability to cases where $N > cr$.

To address this limitation, we propose a novel graphical method for visual assessment of matrix-variate normality, called the Matrix Healy (MHealy) plot, an extension of the Healy plot\citep{healy1968multivariate} for vector-valued data. The main contributions of this paper are as follows.

 \begin{enumerate}[(i)]
 \item MHealy plot provides a practical graphical tool for assessing matrix-variate normality, overcoming the limitation of the DD plot, which can not be applied when $N<cr$. Furthermore, empirical results show that the MHealy plot offers superior visual assessment even when $N<cr$.  
 
 \item Theoretically, we show that the matrix-based MSD is more accurate than the vector-based MSD when a Kronecker product covariance structure exists. Since the DD plot relies on the vector-based MSD, whereas the MHealy plot is purely based on the matrix-based MSD, our theoretical findings substantiate the superiority of the MHealy plot over the DD plot.
 \end{enumerate}  

 The remainder of this paper is organized as follows. \refs{sec:pre} provides an overview of the preliminary knowledge. \refs{sec:meth} reviews the existing method for visualizing matrix-variate normality and introduces the proposed method in detail. \refs{sec:sim} validates the effectiveness of the proposed method using simulation datasets, whereas \refs{sec:appli} demonstrates its application on real-world data. Finally, \refs{sec:conc} discusses potential future directions and areas for further research and development.

\section{Preliminary}\label{sec:pre}
This section provides a concise overview of the multivariate normal distribution, the matrix-variate normal distribution, and the existing method of assessing matrix-variate normality testing. The notations used consistently throughout the paper include the transpose of a vector $\bx$ or matrix $\bX$, denoted by $\bx^{\top}$ or $\bX^{\top}$, respectively; the trace of a matrix, indicated by $\tr(\cdot)$; the vectorization operation, represented by $\vc(\cdot)$; and the Kronecker product, signified by $\otimes$. 

\subsection{Multivariate normal distribution and its Mahalanobis squared distance}
\label{sec:mnandmsd}
\subsubsection{Multivariate normal distribution }
Suppose that a $d$-dimensional random vector $\bx$ follows the multivariate normal distribution with center $\bmu\in \R^{d}$, positive definite matrix $\bSig\in \R^{d\times d}$, denoted by $\bx\sim\cN_d(\bmu,\bSig)$, then the probability density function (p.d.f.) of $\bx$ is given by
\begin{IEEEeqnarray}{rCl}\label{eqn:mn.pdf}
p(\bx)&=&(2\pi)^{-\frac{d}{2}}|\bSig|^{-\frac{1}{2}} \exp \left\{-\frac{1}{2}(\bx-\bmu)^{\top}\bSig^{-1}(\bx-\bmu) \right\}.
\end{IEEEeqnarray}

Consider $N$ $d$-dimensional vectors $\bx_1, \ldots, \bx_N$, where each $\bx_n$ represents a realisation from a multivariate random variable $\bx$. Then the maximum likelihood estimates(MLEs) for $\bmu$ and $\bSig$  are given by \citep{anderson1985maximum}
\begin{IEEEeqnarray}{rCl} 
\bwmu &=& \frac{1}{N}\sum_{n=1}^{N}\bx_n,\label{eqn:mn.mle.mu} \\
\bwSig&=& \frac{1}{N-1}\sum_{n=1}^{N}(\bx_n-\bwmu)(\bx_n-\bwmu)^{\top}.\label{eqn:mn.mle.sig}
\end{IEEEeqnarray} 

\subsubsection{ MSD of Multivariate normal distribution} \label{sec:msd.mn}
 The vector-based MSD, given $\bx_n, \bmu$ and $\bSig$, is defined as
\begin{IEEEeqnarray}{rCl}\label{eqn:mn.msd}
\Delta(\bx_n,\bmu,\bSig)&=&(\bx_n-\bmu)^{\top}\bSig^{-1}(\bx_n-\bmu).
\end{IEEEeqnarray}

It is a well-established fact that \citep{mardia1979multivariate}  
\begin{IEEEeqnarray}{rCl} \label{eqn:mn.msd.chis}
\Delta(\bx_n,\bmu,\bSig) \sim \chi^2_{d},
\end{IEEEeqnarray}
where $\chi^2_{d} $ denotes a chi-square distribution with $d$ degrees of freedom. 

\subsection{Matrix-variate normal distribution and its MSD}\label{sec:mvnandmsd}
\subsubsection{Matrix-variate normal distribution}
\label{sec:mvn}
Suppose that a random matrix $\bX\in\R^{c\times r}$ is said to follow a matrix-variate normal distribution with center matrix $\bM\in\R^{c\times r}$, column and row covariance matrices $\bSig_c\in\R^{c\times c}$ and $\bSig_r\in\R^{r\times r}$, denoted by $\bX\sim \cN_{c, r}\left(\bM, \bSig_c, \bSig_r\right)$, if $\vc(\bX)\sim \cN_{cr}(\vc(\bM), \bSig_r\otimes\bSig_c)$ \citep{gupta1999matrix}. The p.d.f. of $\bX$ is given by 
\begin{IEEEeqnarray}{rCl}
p\left(\bX\right)&=& (2\pi)^{-\frac{c r}{2}}\left|\bSig_r\right|^{-\frac{c}{2}}\left|\bSig_c\right|^{-\frac{r}{2}}\etr\left\{-\frac{1}{2}\bSig_c^{-1}(\bX-\bM)\bSig_r^{-1}(\bX-\bM)^{\top}\right\},\label{eqn:mvn.pdf}
\end{IEEEeqnarray}
where $\etr(\cdot)=\exp(\tr(\cdot))$.

Note that there is an identifiability issue regarding the parameters $\bSig_c$ and $\bSig_r$. Specifically, if $a$ is any positive constant, the following equation holds:
\begin{IEEEeqnarray}{rCl}\label{eqn:mvn.kron}
\left(\frac{1}{a}\bSig_r\right) \otimes(a \bSig_c) &=& \bSig_r \otimes \bSig_c,
\end{IEEEeqnarray}
this implies that replacing $\bSig_c$ and $\bSig_r$ with  $\frac{1}{a} \bSig_c$ and $a \bSig_r$, respectively, does not affect the probability density function defined in  (\ref{eqn:mvn.pdf}). To address this identifiability issue, several solutions have been proposed in the literature, such as constraining  $\tr(\bSig_c)=c$ or setting the first element of $\bSig_c$ to ${\bSig}_{c,11}=1$ \citep{fosdick2014separable, gallaugher2018finite}.

Consider $N$ $c \times r$ matrices $\bX_1, \ldots, \bX_N$, where each $\bX_n$ is a realization of the matrix-variate normal random variable $\bX$. The maximum likelihood estimate of $\bM$ is given by
\begin{IEEEeqnarray}{rCl}\label{eqn:mvn.mle.M}
\bwM &=& \frac{1}{N}\sum_{n=1}^{N}\bX_n.
\end{IEEEeqnarray}

The maximum likelihood estimates for $\bSig_c$ and $\bSig_r$ are obtained by alternately updating $\bSig_c$ and $\bSig_r$ until convergence is achieved \citep{dutilleul1999mle}:
\begin{IEEEeqnarray}{rCl}
\bwSig_c&=& \frac{1}{Nr}\sum_{n=1}^{N}(\bX_n-\bwM)\bwSig_r^{-1}(\bX_n-\bwM)^{\top},\label{eqn:mvn.mle.sigc} \\
\bwSig_r &=& \frac{1}{Nc}\sum_{n=1}^{N}(\bX_n-\bwM)^{\top}\bwSig_c^{-1}(\bX_n-\bwM).\label{eqn:mvn.mle.sigr}
\end{IEEEeqnarray}

\subsubsection{MSD of Matrix-variate normal distribution }
\label{sec:mvn.msd}
The matrix-based MSD can then be written as\citep{povcuvca2023visual}
\begin{IEEEeqnarray}{rCl}\label{eqn:mvn.msd}
\Delta_{M}(\bX_n,\bM, \bSig_c, \bSig_r)=\tr\left\{\bSig_c^{-1}\left(\bX_n-\bM\right)\bSig_r^{-1}\left(\bX_n-\bM\right)^{\top}\right\}.
\end{IEEEeqnarray}

It is recognized that\citep{gupta1999matrix}
\begin{IEEEeqnarray}{rCl} \label{eqn:mvn.msd.chis}
\Delta_{M}(\bX_n,\bM, \bSig_c, \bSig_r)\sim \chi^2_{cr}.
\end{IEEEeqnarray}

\subsection{Distance-Distance plot}\label{sec:ddplot}
Recall the relationship between the matrix normal and the multivariate normal distributions in \refs{sec:mvn}. Let $\bmu=\vc(\bM)$ and $\bSig=\bSig_r \otimes \bSig_c$. It follows that $\vc(\bX_n)\sim \cN_{cr}(\bmu, \bSig)$, Then
\begin{IEEEeqnarray}{rCl}
\Delta(\vc(\bX_n),\bmu,\bSig)&=&(\vc(\bX_n)-\bmu)^{\top}\bSig^{-1}(\vc(\bX_n)-\bmu).
\label{eqn:mvn.vcmsd}
\end{IEEEeqnarray}

Moreover, the maximum likelihood estimates of $\bmu$ and $\bSig$ are, respectively:
\begin{IEEEeqnarray}{rCl}
\bwmu &=& \frac{1}{N}\sum_{n=1}^{N}\vc(\bX_n)=\vc(\bwM),\label{eqn:mvn.mle.vcmu} \\
\bwSig &=& \frac{1}{N-1}\sum_{n=1}^{N}\left(\vc(\bX_n)-\bwmu\right)\left(\vc(\bX_n)-\bwmu\right)^{\top}.\label{eqn:mvn.mle.vcsig} 
\end{IEEEeqnarray}

 \cite{povcuvca2023visual} proposed a graphical tool for visually assessing matrix-variate normality. This method calculates two Mahalanobis squared distances, the vector-based MSD, $\Delta(\vc(\bX_n),\bmu,\bSig)$ and the matrix-based MSD, $\Delta_{M}(\bX_n,\bM,\bSig_c,\bSig_r)$. These distances are then presented in a scatter plot, called the DD plot, which highlights the consistency between matrix-based and vector-based MSD. If the points in the DD plot align with the reference line $y=x$, the assumption of matrix-variate normality is supported. However, the vector-based MSD used in the DD plot require vectorizing matrix observations, which makes MLEs infeasible when $N<cr$. 
 
As shown in \reff{fig:matrix_multi_N>d}(b), the DD plot for simulated data from a matrix-variate normal distribution with $N = 1000, c = 2, r = 2$, shows the MSDs aligning closely with the reference line(the red solid line), indicating conformity to the matrix-variate normality assumption. However, \reff{fig:matrix_multi_N close d}(b) presents data simulated under the same distribution with $N = 1000, c = 30, r = 30$. In this case, the MSDs exhibit substantial variability and deviate significantly from the reference line. This suggests that the data fail to meet the matrix-variate normality assumption despite being generated from a matrix-variate normal distribution. This highlights the limitations of the DD plot in reliably assessing matrix-variate normality. To overcome this limitation, the following section introduces the proposed visualization method.

\section{Methodology}\label{sec:meth}
\subsection{The proposed Matrix Healy plot}\label{sec:mtplot}

The Healy plot, originally introduced by \cite{healy1968multivariate}, provides a visualization method to evaluate multivariate normality by plotting ordered Mahalanobis squared distances against the expected order statistics of a chi-square distribution. When the data conform to a multivariate normal distribution, the points should approximately align with the line $y=x$. Subsequently, \cite{lin2015robust} extended this approach to evaluate the goodness of fit for factor analysis models. Their method, referred to as the Healy-type plot, involves plotting the cumulative probabilities of the $F(d,\nu)$ distribution corresponding to the ordered MSD $f_{n},n=1,2,\ldots, N$, against their nominal values $1/N, 2/N, \ldots, 1$. Here $F(d,\nu)$  represents the $F$ distribution with degrees of freedom $d$ and $\nu$.


Building on these earlier studies, this paper proposes a practical tool for assessing matrix-variate normality testing, termed the matrix Healy (MHealy) plot. By fully considering the matrix structure of the data, this method eliminates the sample size restriction. The specific process for implementing this newly developed visualization method to assess matrix-variate normality is outlined as follows:
\begin{enumerate}[(i)]
\item Calculate the matrix-based MSD, $\Delta_{M}(\bX_n,\bwM, \bwSig_c, \bwSig_r),n=1,2,\ldots, N$ by (\ref{eqn:mvn.msd}). 

\item Compute cumulative probability $p_{n}$ of the chi-square distribution associated with the ordered value of $\Delta_{M}(\bX_n,\bwM, \bwSig_c, \bwSig_r)$, $n=1,2,\ldots, N$. Specifically, $p_{n}=P(X \leq \Delta_{M}(\bX_n,\bwM, \bwSig_c, \bwSig_r) )$, where $X \sim \chi^2_{cr}$ and $\chi^2_{cr}$ denotes the chi-square distribution with degrees of freedom $cr$, 
as outlined in the method described by \cite{lin2015robust}. 

\item Construct a scatter plot using nominal values $1/N, 2/N, \ldots, 1$ as the horizontal axis and the cumulative probabilities as the vertical axis.
\end{enumerate}

  If the data follow the assumption of matrix-variate normality, the observed MSD values are consistent with the theoretical chi-square distribution\citep{van2000asymptotic}. Therefore, the resulting plot should approximate a straight line $y=x$ (the reference line). In contrast, significant deviations of the points from this 45-degree line indicate stronger evidence against the matrix-variate normality assumption.

\subsection{The Properties of MSD} \label{sec:prop}

\begin{prop} \label{prop:mvn.msd1}
	If a Kronecker product structure exists for $\bSig=\bSig_r \otimes \bSig_c$, and let us assume $c=r$, the following results hold:

\begin{IEEEeqnarray}{rcl}
\Delta(\vc(\bX_n),\bwmu,\bwSig)
&=& \Delta(\vc(\bX_n),\bmu,\bSig)+O_p\left(\frac{c^2}{\sqrt{N}}\right) \label{eqn:mn.msdpro1}, \\
 \Delta_{M}(\bX_n,\bwM, \bwSig_c, \bwSig_r) 
&=& \Delta_{M}(\bX_n,\bM, \bSig_c, \bSig_r)+ O_p\left(\frac{c}{\sqrt{N}}\right).
\label{eqn:mvn.msdpro1}
\end{IEEEeqnarray}
\end{prop}

\begin{proof}
	The proof is provided in the Appendix. 
\end{proof}

 Based on \refp{prop:mvn.msd1}, the error growth rate for the vectorized Mahalanobis squared distance is $c$ times greater than that of the matrix-based Mahalanobis distance. These findings highlight the superior accuracy and reliability of the proposed matrix-based Mahalanobis squared distances for analytical applications. 

\subsection{The relationship and difference among MHealy plot, Healy type plot and DD plot} \label{sec:relation}

All three plots are grounded in Mahalanobis squared distances and are used to evaluate whether the observed data adhere to normality. The Healy type plot is only applicable to vector-valued data, utilizing vector-based MSD and mapping these values to chi-square cumulative probabilities to assess multivariate normality. In contrast, the MHealy plot extends this concept to matrix-valued data, mapping matrix-based MSD values to chi-square cumulative probabilities to directly evaluate matrix-variate normality. Similarly, the DD plot is designed for matrix-valued data but focuses on comparing the relationship between matrix-based and vector-based MSD. 

As a result, although both the MHealy and DD plots are designed for matrix-valued data, a key distinction lies in their underlying MSD: the DD plot relies on vector-based MSD, while the MHealy plot exclusively employs the more accurate matrix-based MSD, as shown in \refp{prop:mvn.msd1}.

\section{ Simulations}\label{sec:sim}

In this section, we conduct experiments on simulated datasets to compare the performance of our proposed Mhealy plot with DD plot and two-phase approach for assessing whether the data conform to the matrix-variate normal distribution. Note that we use Healy type plot to verify multivariate normality in the first phase of two-phase approach. Since the existing approach is not feasible when the sample size $N$ is smaller than the dimension of the vectorized data $d=cr$, we conducted separate experiments for cases where $N > d$, $N$ is close to $d$, and $N < d$. The study is implemented on a notebook with MATLAB 2022b, an Intel(R) Core(TM) i7-11800H 2.30 GHz CPU, and 16 GB RAM. 


In these experiments, data are generated from a matrix-variate normal distribution as defined in (\ref{eqn:mvn.msd}) and a strict multivariate normal distribution without a Kronecker product structure as described in (\ref{eqn:mn.msd}). The sample size is set to $N=1000$,
with dimensions $d \in \{4, 900, 1600, 10000, 40000\}$. For each scenario, the mean vector and covariance matrix are randomly generated. Importantly, square matrices are used throughout the experiment: for $d = 4$, the matrix-value data corresponds to a $c \times r = 2 \times 2$ matrix; for $d = 900$, to a $c \times r = 30 \times 30$ matrix; for $d = 1600$, to a $c \times r = 40 \times 40$ matrix; for $d =10000$, to a $c \times r = 100 \times 100$ matrix and for $d = 40000$, to a $c \times r =200 \times 200$ matrix. 

\subsection{The sample size $N$ is greater than the dimension $d$}
\label{sec: N greater d}
In this experiment, we investigate the performance of the proposed Mhealy plot, DD plot, and two-phase approach for assessing matrix-variate normality when $N>d$. \reff{fig:matrix_multi_N>d}(a,b,c) illustrates the MHealy plot (left panel), DD plot (middle panel), and Healy type plot (right panel) for simulated matrix-variate normal datasets with $N=1000,c=2,r=2$. The points in MHealy plot and DD plot align closely with the reference line(the red solid line), indicating that the data conform to the matrix-variate normality assumption when $N>d=cr$. In particular, the equivalence between multivariate and matrix-variate normal distributions, as discussed in \refs{sec:mvn}, implies that if vectorized data do not exhibit multivariate normality, then the original matrices cannot follow a matrix-variate normal distribution. Therefore, Healy type plot is expected to align with the reference line. The points of Healy type plot in \reff{fig:matrix_multi_N>d}(c)  align with the reference line, suggesting that Healy type plot effectively assesses multivariate normality for the vectorized data in the first phase of two-phase approach.

\begin{figure*}[htb]
	\centering \scalebox{0.74}[0.74]{\includegraphics*{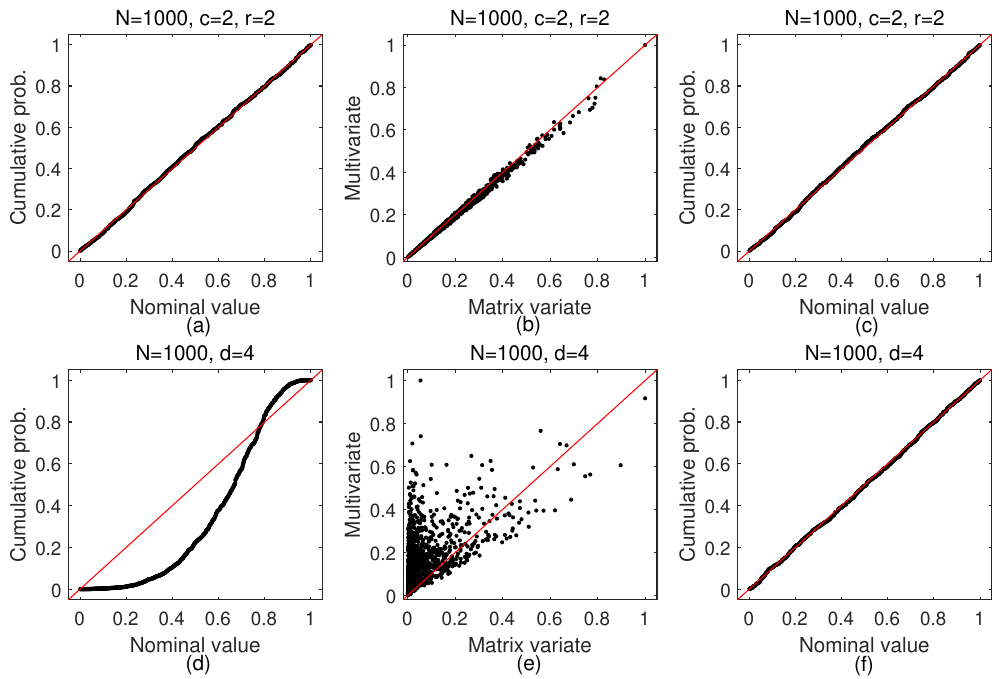}}
	\caption{Top row: MHealy plot(left panel), DD plot(middle panel) and Healy type plot(right panel) for simulated matrix-variate normal datasets with $N=1000,c=2,r=2$, where a matrix-variate normal structure is present. Bottom row: MHealy plot(left panel), DD plot(middle panel) and Healy type plot(right panel) for simulated strict multivariate normal datasets with $N=1000,d=4$, where a matrix-variate normal structure is absent.} \label{fig:matrix_multi_N>d}
\end{figure*}

\reft{tab:fig1lrt} presents the $p$-values of LRT for simulated matrix-variate normal and strict multivariate normal datasets at the $5\%$ significance level. These datasets correspond to those analyzed in \reff{fig:matrix_multi_N>d}. For data generated from a matrix-variate normal with $N=1000,c=2,r=2$, the LRT yields a $p$-value of 0.156. As this value exceeds the significance threshold of 0.05, there is insufficient evidence to reject the null hypothesis of a Kronecker structure covariance matrix in the second phase of two-phase approach. By combining the results from the first and second stages, the two-phase approach effectively assess whether the data follow matrix-variate normality when $N=1000,c=2,r=2$.

\begin{table}[htbp]
	\centering
	\caption{\label{tab:fig1lrt} The $p$-values of LRT for simulated matrix-variate normal and  strict multivariate normal datasets at the $5\%$ significance level.}
	\begin{tabular}{ccc}
		\toprule 
	distribution	  &datasets            &  $p$ value   \\ \midrule	
	matrix-variate normal	 &$N=1000,c=2,r=2$     &  0.156\\ 
		                   &$N=1000,c=30,r=30$   &  0 \\ 
    strict multivariate normal  & $N=1000,d=4$          &  0\\
                              & $N=1000,d=900$ &  0  \\          
  \bottomrule
	\end{tabular}
\end{table}

\reff{fig:matrix_multi_N>d}(d,e,f) shows the MHealy plot (left panel), DD plot (middle panel) and Healy  type plot (right panel) for strict multivariate normal datasets with $N=1000,d=4$. The points in the MHealy plot and DD plot deviate significantly from the reference line, indicating that the data do not follow the matrix-variate normality assumption. This finding is consistent with the data's origin, as it was generated from strict multivariate normal distributions rather than matrix-variate normal distributions. Moreover, the points in the Healy type plot align closely with the reference line, indicating that the Healy type plot appropriately assesses multivariate normality. The LRT yields a $p$-value of 0, as shown in \reft{tab:fig1lrt}. Since this $p$-value is less than the significance threshold of 0.05, there is sufficient evidence to reject the null hypothesis of a Kronecker structure covariance matrix. Therefore, the two-phase approach effectively assesses whether the data follow matrix-variate normality for $N=1000$ and $d=4$.

\subsection{The sample size $N$ is close to the dimension $d$}
\label{sec:N close d}
In this experiment, we investigate the performance of the proposed Mhealy plot, DD plot, and two-phase approach for assessing matrix-variate normality when $N$ is close to $d$. \reff{fig:matrix_multi_N close d}(a,b,c) illustrates the MHealy plot (left panel), DD plot (middle panel), and Healy type plot (right panel) for simulated matrix-variate normal datasets with $N=1000,c=30,r=30$. In the MHealy plot(\reff{fig:matrix_multi_N close d}(a)), the points align closely with the reference line, indicating that the data adhere to the matrix-variate normality assumption when $N$ is close to $d$, which is consistent with the data's origin. Conversely, the DD plot (Figure \ref{fig:matrix_multi_N close d}(b)) shows significant deviations from the reference line, providing evidence against matrix-variate normality. Similarly, the Healy type plot (Figure \ref{fig:matrix_multi_N close d}(c)) exhibits noticeable deviations from the reference line, indicating that it is ineffective in evaluating multivariate normality for the vectorized data in this scenario. The LRT yields a $p$-value of 0, as reported in \reft{tab:fig1lrt}. Since this $p$-value is below the significance threshold of 0.05, there is strong evidence to reject the null hypothesis of a Kronecker structure covariance matrix. Thus, the DD plot and the two-phase approach are not suitable for assessing matrix-variate normality when $N$ is close to $d$.

\reff{fig:matrix_multi_N close d}(d,e,f) shows the MHealy plot (left panel), DD plot (middle panel) and Healy type plot (right panel) for strict multivariate normal datasets with $N=1000,d=900$. The points in the MHealy plot and DD plot deviate significantly from the reference line, indicating that the data do not follow the matrix-variate normality assumption. This observation aligns with the data's origin, as it was generated from strict multivariate normal distributions. Additionally, the Healy type plot shows significant deviations from the reference line, indicating its ineffectiveness in assessing multivariate normality. The LRT yields a $p$-value of 0, as shown in \reft{tab:fig1lrt}. As the $p$-value is below the significance level of 0.05, the null hypothesis can be rejected. Consequently, the two-phase approach proves to be ineffective in determining whether the data follow matrix-variate normality when $N$ is close to $d$.

\begin{figure*}[htb]
	\centering \scalebox{0.74}[0.74]{\includegraphics*{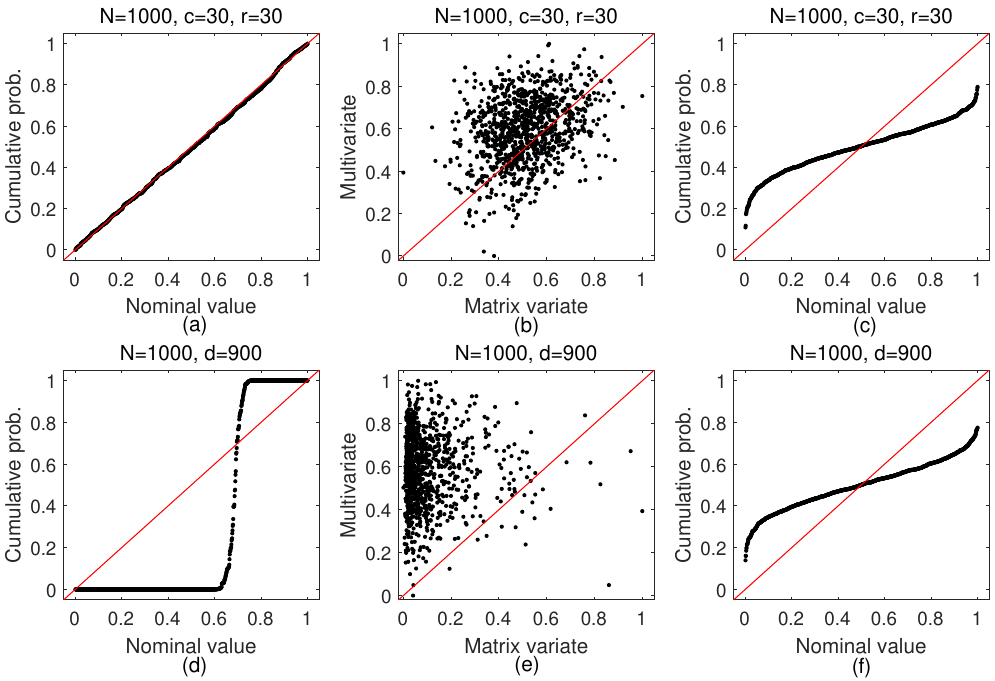}}
	\caption{Top row: MHealy plot(left panel), DD plot(middle panel) and Healy type plot(right panel) for simulated matrix-variate normal datasets with $N=1000,c=30,r=30$, where a matrix-variate normal structure is present. Bottom row: MHealy plot(left panel), DD plot(middle panel) and Healy type plot(right panel) for simulated strict multivariate normal datasets with $N=1000,d=900$, where a matrix-variate normal structure is absent.}  \label{fig:matrix_multi_N close d}
\end{figure*}

\subsection{The sample size $N$ is smaller than the dimension $d$}
\label{sec: N small d}
In this experiment, we assess the performance of the proposed Mhealy plot for evaluating matrix-variate normality when $N<d$. It is important to note that the DD plot and two-phase approach are not feasible under these conditions, as the parameter estimates for the multivariate normal distribution become unreliable. \reff{fig:matrix_multi1000} displays the MHealy plots for simulated matrix-variate normal data(Top row) and strictly multivariate normal data(Bottom row) under the condition $N<d$. 
In \reff{fig:matrix_multi1000}(a,b,c), with sample sizes $N=1000$ and dimensions $d$ increasing from 1600 to 40000, the scatter points in the MHealy plots closely align with the reference line, indicating conformity to the assumption of matrix-variate normal distribution. Conversely, \reff{fig:matrix_multi1000}(d,e,f) illustrates the MHealy plots for strictly multivariate normal data, where the scatter points deviate significantly from the reference line. These deviations reflect the lack of adherence to the matrix-variate normality assumption, which can be attributed to the absence of a Kronecker product covariance structure. These results highlight that the MHealy plot is a robust and practical tool for detecting matrix-variate normality, even in high-dimensional scenarios where $N < d$.

\begin{figure*}[htb]
	\centering \scalebox{0.74}[0.74]{\includegraphics*{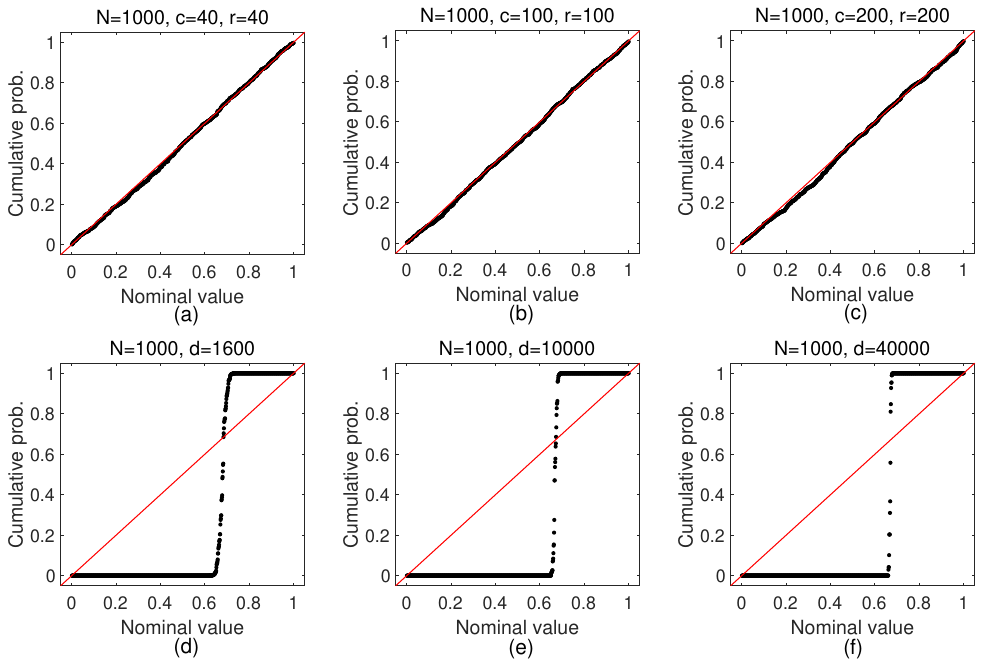}}
	\caption{Top row: MHealy plots for simulated matrix-variate normal datasets for $N=1000$ and $d \in \{1600, 10000, 40000\}$ with randomly selected mean and variance parameters. Bottom row: MHealy plots for simulated strict multivariate normal datasets for $N=1000$ and $d \in \{1600, 10000, 40000\}$, where a matrix-variate normal structure is absent.} \label{fig:matrix_multi1000}
\end{figure*}

\section{Applications}\label{sec:appli}
In this section, we aim to perform a comparative analysis of the performance of the Mhealy plot, DD plot, and two-phase approach on real datasets under two scenarios: $N>d$ and $N<d$. To this end, we selected two representative real datasets to evaluate the effectiveness and applicability of these methods under different conditions.

\subsection{Fama–French 10 by 10 return series}\label{sec:appli.ret}
 We assess matrix-variate normality on the Fama-French 10 by 10 return series used in \cite{wang2019factor}. This dataset includes 624 observations, covering monthly returns from January 1964 to December 2015, over 624 consecutive months. Each observation forms a $10 \times 10$ matrix, with columns representing ten levels sorted by increasing market capitalization and rows representing ten levels sorted by increasing book-to-market ratio. Following the methodology of \cite{wang2019factor}, we adjusted each observation by subtracting the corresponding monthly excess market return from the original returns, yielding a series of market-adjusted returns. 
 
 \reff{fig:trx} presents the MHealy plot and DD plot of the return series. Both plots reveal that the scatters deviate significantly from the reference line, showing substantial variability. This observation suggests that the data does not follow a matrix-variate normal distribution, offering no evidence to support the assumption of matrix-variate normality.

To make a clearer comparison between the two plots, we generate data from a matrix-variate normal distribution with $N=150$, $c=10$, and $r=10$, using parameter estimates derived from the return series under the matrix-variate assumption. 

As illustrated in \reff{fig:trx_add}, the scatter points in the MHealy plot (\reff{fig:trx_add}(a)) align closely with the reference line, providing strong support for the matrix-variate normality assumption. In contrast, the scatter points in both the DD plot(\reff{fig:trx_add}(b)) and the Healy type plot(\reff{fig:trx_add}(c)) deviate from the reference line, with greater variability observed in the DD plot compared to the Healy type plot. This suggests that the DD plot violates the matrix-variate normality assumption. Similarly, the Healy type plot demostrates the vectorized data do not follow the multivariate normality assumption, indicating a failure in the first phase of the two-phase approach. Furthermore, the $p$-value of the LRT is 0, which confirms that the second phase of the two-phase approach identified the absence of a Kronecker-structured covariance matrix, a result that is not in accordance with the actual data.


\begin{figure*}[htb]
	\centering \scalebox{0.8}[0.8]{\includegraphics*{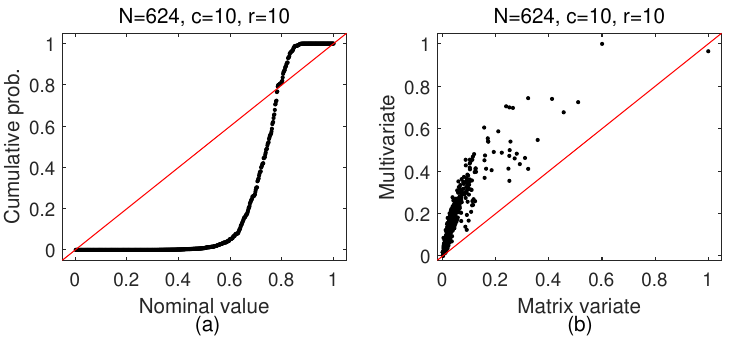}}
	\caption{ MHealy plot(left panel) and DD plot(right panel) for the return series indicate the lack of presence of a matrix-variate normal structure. } \label{fig:trx}
\end{figure*}

\begin{figure*}[htb]
	\centering \scalebox{0.8}[0.8]{\includegraphics*{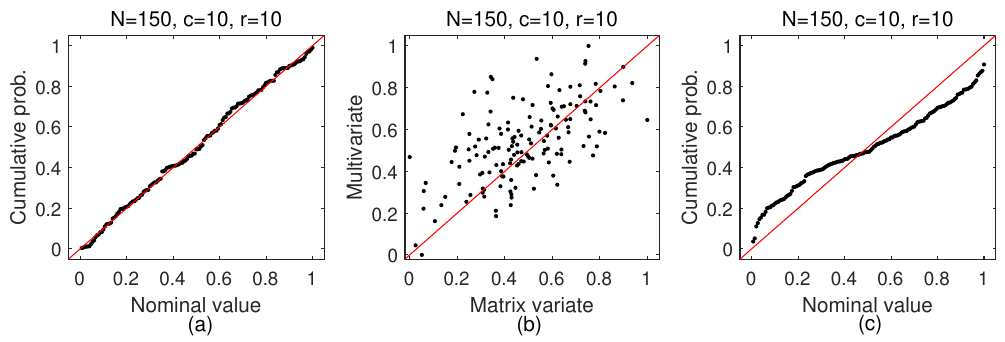}}
	\caption{ MHealy plot(left panel), DD plot(middle panel) and Healy type plot(right panel) for the simulated matrix-variate normal datasets using the estimated parameters of return series. } \label{fig:trx_add}
\end{figure*}

    \subsection{The sensory cheese data}\label{sec:appli.che}
We study the sensory cheese data where the number of observed samples $ N $ is smaller than the dimension of the vectorized data $d = cr $. This data comes from the sensory evaluation of cream cheese, as detailed in \cite{sensorycheese2002}. Specifically, the study involves sensory tests conducted by eight professional evaluators who assessed ten products in three replicates, covering 23 sensory variables in each test, resulting in typical matrix-structured data. Therefore, we have $ N = 30 $, $ c = 8 $, and $ r = 23 $. For more details on the dataset, please refer to \citep{bro2008-multiway-sensory}. Note that the DD plot and the two-phase approach are not feasible when $N<d$. 

\reff{fig:cheese} shows the MHealy plot of the sensory data, which reveals a slight deviation between the scatter points and the reference line, indicating a small degree of variability. This observation may suggest a mild departure from the matrix-variate normal distribution.

\begin{figure*}[htb]
	\centering 
	\scalebox{0.7}[0.65]{\includegraphics*{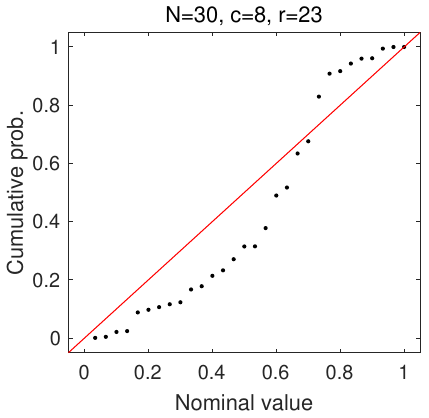}}
	\caption{ MHealy plot for the sensory cheese data indicates the lack of a matrix-variate normal structure. } \label{fig:cheese}
\end{figure*}

\section{Conclusion and future works}\label{sec:conc} 

This paper introduces the Matrix Healy (MHealy) plot, a novel graphical method for matrix-variate normality testing. By leveraging the relationship between cumulative probabilities and nominal values, the MHealy plot based on matrix-based MSD offers an innovative visualization tool with distinct advantages over existing methods, particularly in cases where the sample size is smaller than the dimension of the vectorized data. The proposed method has been theoretically proven to provide a more accurate and reliable framework for analysis.
Empirical validation demonstrates the MHealy plot's accuracy and practicality in assessing the conformity of matrix-valued data to the matrix-variate normal distribution.

Given that real-world datasets often display heavy tails, asymmetry, and skewness, future research will explore the applicability of the MHealy plot to matrix-variate data with these characteristics. In the long term,  this method will be extended to tensor-variate data, addressing the challenges of high-dimensional datasets beyond three dimensions and meeting the analytical demands of advanced data types \citep{kossaifi2020tensor,llosa2022reduced}.

\begin{appendices}

\section{The Proof for Proposition 1   }\label{m.Etau}
	\begin{proof} \label{app:proofs}	
		 For simplicity, let us assume $c=r$. Due to the error of each component $\hat{\mu}_{j}-\mu_{j}=O_p(N^{-\frac{1}{2}})$\citep{van2000asymptotic}, and the presence of $cr$ components, therefore,
       $||\bwmu-\bmu||^2_{2}=(\bwmu-\bmu)'(\bwmu-\bmu)=\sum_{j=1}^{cr}(\hat{\mu}_{j}-\mu_{j})^2=O_p(\frac{cr} {N})$\citep{mehta2004random,wainwright2019high}, Furthermore,
       \begin{IEEEeqnarray}{rCl}
       ||\bwmu-\bmu||_2=O_p(\frac{\sqrt{cr}}{\sqrt{N}})  \label{eqn:par.mu}.
       \end{IEEEeqnarray}
       
      Similarly, it can be obtained that 
       \begin{IEEEeqnarray}{rCl}
       ||\bwSig-\bSig||_F&=&O_p(\frac{cr}{\sqrt{N}}),\label{eqn:par.sig} \\
       ||\bwM-\bM||_F&=&O_p(\frac{\sqrt{cr}}{\sqrt{N}}),\label{eqn:par.M} \\
       ||\bwSig_c-\bSig_c||_F&=&O_p(\frac{c}{\sqrt{Nr}}),\label{eqn:par.sig_c}\\
       ||\bwSig_r-\bSig_r||_F&=&O_p(\frac{r}{\sqrt{Nc}})\label{eqn:par.sig_r}. 
       \end{IEEEeqnarray}
       Perform first-Order approximation \citep{wainwright2019high,meyer2023matrix} for $\bwSig^{-1}$, $\bwSig_c^{-1}$ and $\bwSig_r^{-1}$, respectively, resulting in the following expressions:
       \begin{IEEEeqnarray}{rCl}
        \bwSig^{-1} &\approx& \bSig^{-1}- \bSig^{-1} E_{\bwSig}\bSig^{-1}, \label{eqn:par.Isig} \\
        \bwSig_c^{-1} &\approx& \bSig_c^{-1}- \bSig_c^{-1} E_{\bwSig_c}\bSig_c^{-1}, \label{eqn:par.Isig_c} \\
        \bwSig_r^{-1} &\approx& \bSig_r^{-1}- \bSig_r^{-1} E_{\bwSig_r}\bSig_r^{-1} , \label{eqn:par.Isig_r}
       \end{IEEEeqnarray}
	 where $E_{\bwSig}=\bwSig-\bSig$, $E_{\bwSig_c}=\bwSig_c-\bSig_c$, $E_{\bwSig_r}=\bwSig_r-\bSig_r$, and
      $||E_{\bwSig}||_F=O_p(\frac{cr}{\sqrt{N}})$, $||E_{\bwSig_c}||_F =O_p(\frac{c}{\sqrt{Nr}})$,
       and $||E_{\bwSig_r}||_F =O_p(\frac{r}{\sqrt{Nc}})$. 

Thus, expanding $\Delta(\vc(\bX_n),\bwmu,\bwSig)$, and substituting  \refe{eqn:par.Isig} into \refe{eqn:mnmsd1}, we can derive

\begin{IEEEeqnarray}{rCl}
		\nonumber
		&& \Delta(\vc(\bX_n),\bwmu,\bwSig) \\
		\nonumber
		&=& \left(\vc(\bX_n)-\bwmu \right)^{\top}\bwSig^{-1}\left(\vc(\bX_n)-\bwmu\right)\\
		\nonumber
		&=& \left(\vc(\bX_n)-\bmu-(\bwmu-\bmu)\right)^{\top}\bwSig^{-1}\left(\vc(\bX_n)-\bmu-(\bwmu-\bmu)\right)\label{eqn:mnmsd1}\\
		&=& (\vc(\bX_n)-\bmu)^{\top}\bwSig^{-1}(\vc(\bX_n)-\bmu)-2(\vc(\bX_n)-\bmu)^{\top}\bwSig^{-1}(\bwmu-\bmu) +\> (\bwmu-\bmu)^{\top}\bwSig^{-1}(\bwmu-\bmu)  \label{eqn:mnmsd2}\\
		\nonumber
         &=& (\vc(\bX_n)-\bmu)^{\top}\left(\bSig^{-1}- \bSig^{-1} E_{\bwSig}\bSig^{-1}\right)(\vc(\bX_n)-\bmu)-2(\vc(\bX_n)-\bmu)^{\top}\left(\bSig^{-1}- \bSig^{-1} E_{\bwSig}\bSig^{-1}\right)(\bwmu-\bmu)\\
        && +\> (\bwmu-\bmu)^{\top}\left(\bSig^{-1}- \bSig^{-1} E_{\bwSig}\bSig^{-1}\right)(\bwmu-\bmu) \label{eqn:mnmsd3}.
  \end{IEEEeqnarray}

Ignoring second-order term,  \refe{eqn:mnmsd3} can be approximated
\begin{IEEEeqnarray}{rCl}
	\nonumber
	&& \Delta(\vc(\bX_n),\bwmu,\bwSig) \\
 \nonumber
  &\approx & (\vc(\bX_n)-\bmu)^{\top}\left(\bSig^{-1}- \bSig^{-1} E_{\bwSig}\bSig^{-1}\right)(\vc(\bX_n)-\bmu)-2(\vc(\bX_n)-\bmu)^{\top}\left(\bSig^{-1}- \bSig^{-1} E_{\bwSig}\bSig^{-1}\right)(\bwmu-\bmu)    \\
  \nonumber
  &=& (\vc(\bX_n)-\bmu)^{\top}\bSig^{-1}(\vc(\bX_n)-\bmu)- (\vc(\bX_n)-\bmu)^{\top}\bSig^{-1} E_{\bwSig}\bSig^{-1}(\vc(\bX_n)-\bmu)\\
  \nonumber
   &&-2(\vc(\bX_n)-\bmu)^{\top}\bSig^{-1}(\bwmu-\bmu)+2(\vc(\bX_n)-\bmu)^{\top}\bSig^{-1} E_{\bwSig}\bSig^{-1}(\bwmu-\bmu)   \\
   \nonumber
  &=&  \Delta(\vc(\bX_n),\bmu,\bSig)
         +O_p\left(\frac{cr}{\sqrt{N}} \right)+O_p\left(\frac{\sqrt{cr}}{\sqrt{N}}\right)  +O_p\left(\frac{(cr)^{\frac{3}{2}}}{N}\right) \\
 &\approx &  \Delta(\vc(\bX_n),\bmu,\bSig)+O_p\left(\frac{cr}{\sqrt{N}} \right)\label{eqn:mnmsdapp1}\\
 \nonumber
&=&  \Delta(\vc(\bX_n),\bmu,\bSig)+O_p\left(\frac{c^2}{\sqrt{N}} \right),
\end{IEEEeqnarray}
where when $N$ is sufficiently large, \refe{eqn:mnmsdapp1} holds.

Similarly, expanding $\Delta_{M}(\bX_n,\bwM,\bwSig_c,\bwSig_r)$,  substituting \refe{eqn:par.Isig_c} and \refe{eqn:par.Isig_r} into \refe{eqn:mvnmsd1},  we can get
	\begin{IEEEeqnarray}{rCl}
		\nonumber
		&& \Delta_{M}(\bX_n,\bwM,\bwSig_c,\bwSig_r) \\
		\nonumber
		& =& \tr\left\{\bwSig_c^{-1}(\bX_n-\bwM)\bwSig_r^{-1}(\bX_n-\bwM)^{\top}\right\} \\
		\nonumber
		& =&  \tr\left\{ \bwSig_c^{-1}\left((\bX_n-\bM)-(\bwM-\bM)\right)\bwSig_r^{-1}\left((\bX_n-\bM)-(\bwM-\bM)\right)^{\top}\right\}\\
        \nonumber
		& =& \tr\left\{\bwSig_c^{-1}(\bX_n-\bM)\bwSig_r^{-1}(\bX_n-\bM)^{\top}\right\}-2\tr\left\{ \bwSig_c^{-1}(\bwM-\bM)\bwSig_r^{-1}(\bX_n-\bM)^{\top}\right\}\\
            && +\>\tr\left\{\bwSig_c^{-1}(\bwM-\bM)\bwSig_r^{-1}(\bwM-\bM)^{\top}\right\} \label{eqn:mvnmsd1}\\
         \nonumber
		& =& \tr\left\{\left(\bSig_c^{-1}- \bSig_c^{-1} E_{\bwSig_c}\bSig_c^{-1} \right) \left(\bX_n-\bM \right)\left(\bSig_r^{-1}- \bSig_r^{-1} E_{\bwSig_r}\bSig_r^{-1}\right)(\bX_n-\bM)^{\top}\right\}\\
        \nonumber
         &&-\> 2\tr\left\{ \left(\bSig_c^{-1}- \bSig_c^{-1} E_{\bwSig_c}\bSig_c^{-1}\right)(\bwM-\bM)\left(\bSig_r^{-1}- \bSig_r^{-1} E_{\bwSig_r}\bSig_r^{-1}\right)(\bX_n-\bM)^{\top}\right\}\\
            && +\>\tr\left\{\left(\bSig_c^{-1}- \bSig_c^{-1} E_{\bwSig_c}\bSig_c^{-1}\right)(\bwM-\bM)\left(\bSig_r^{-1}- \bSig_r^{-1} E_{\bwSig_r}\bSig_r^{-1}\right)(\bwM-\bM)^{\top}\right\}  \label{eqn:mvnmsd2} 
\end{IEEEeqnarray}

Ignoring second-order term,  \refe{eqn:mvnmsd2} can be approximated
\begin{IEEEeqnarray}{rCl}
    \nonumber
	&& \Delta_{M}(\bX_n,\bwM,\bwSig_c,\bwSig_r) \\
    \nonumber
  &\approx & \tr\left\{\left(\bSig_c^{-1}- \bSig_c^{-1} E_{\bwSig_c}\bSig_c^{-1} \right) \left(\bX_n-\bM \right)\left(\bSig_r^{-1}- \bSig_r^{-1} E_{\bwSig_r}\bSig_r^{-1}\right)(\bX_n-\bM)^{\top}\right\}\\
        \nonumber
         &&-\> 2\tr\left\{ \left(\bSig_c^{-1}- \bSig_c^{-1} E_{\bwSig_c}\bSig_c^{-1}\right)(\bwM-\bM)\left(\bSig_r^{-1}- \bSig_r^{-1} E_{\bwSig_r}\bSig_r^{-1}\right)(\bX_n-\bM)^{\top}\right\} \\
  \nonumber
  & =& \tr\left\{\bSig_c^{-1} \left(\bX_n-\bM \right)\bSig_r^{-1}(\bX_n-\bM)^{\top}\right\}
    -\tr\left\{\bSig_c^{-1} \left(\bX_n-\bM \right)\bSig_r^{-1} E_{\bwSig_r}\bSig_r^{-1}(\bX_n-\bM)^{\top}\right\}\\
     \nonumber
    &&-\>\tr\left\{ \bSig_c^{-1} E_{\bwSig_c}\bSig_c^{-1}  \left(\bX_n-\bM \right)\bSig_r^{-1}(\bX_n-\bM)^{\top}\right\}
    +\tr\left\{ \bSig_c^{-1} E_{\bwSig_c}\bSig_c^{-1}  \left(\bX_n-\bM \right)\bSig_r^{-1} E_{\bwSig_r}\bSig_r^{-1}(\bX_n-\bM)^{\top}\right\}\\
        \nonumber
         &&-\> 2\tr\left\{ \bSig_c^{-1}(\bwM-\bM)\bSig_r^{-1}(\bX_n-\bM)^{\top}\right\} 
         +2\tr\left\{ \bSig_c^{-1}(\bwM-\bM)\bSig_r^{-1} E_{\bwSig_r}\bSig_r^{-1}(\bX_n-\bM)^{\top}\right\} \\
         \nonumber
         &&+\> 2\tr\left\{ \bSig_c^{-1} E_{\bwSig_c}\bSig_c^{-1}(\bwM-\bM)\bSig_r^{-1}(\bX_n-\bM)^{\top}\right\}
         -2\tr\left\{ \bSig_c^{-1} E_{\bwSig_c}\bSig_c^{-1}(\bwM-\bM)\bSig_r^{-1} E_{\bwSig_r}\bSig_r^{-1}(\bX_n-\bM)^{\top}\right\}\\
    \nonumber
    & =&  \Delta_{M}(\bX_n,\bM, \bSig_c, \bSig_r)
   + O_p\left(\frac{r}{\sqrt{Nc}}\right)+O_p\left(\frac{c}{\sqrt{Nr}}\right)+O_p\left(\frac{\sqrt{cr}}{N}\right) \\
 \nonumber
  && +\> O_p\left(\frac{\sqrt{cr}}{\sqrt{N}}\right) +O_p\left(\frac{r^{\frac{3}{2}}}{N}\right) +O_p\left(\frac{c^{\frac{3}{2}}}{N}\right)+O_p\left(\frac{cr}{N^{\frac{3}{2}}}\right) \\ 
   & \approx&  \Delta_{M}(\bX_n,\bM, \bSig_c, \bSig_r)+O_p\left(\frac{\sqrt{cr}}{\sqrt{N}}\right)\label{eqn:mvnmsdapp1}\\
 & =&  \Delta_{M}(\bX_n,\bM, \bSig_c, \bSig_r)+O_p\left(\frac{c}{\sqrt{N}}\right),  
\end{IEEEeqnarray}
where when $N$ is sufficiently large, \refe{eqn:mvnmsdapp1} holds.
The proof is concluded.  
\end{proof}
\end{appendices}
\bibliography{journall,jhzhao-pub,lit}
\bibliographystyle{elsarticle-harv}
	
\end{document}